\documentclass[conference,letterpaper]{IEEEtran}
\usepackage{afterpage}
\usepackage[top=72pt, bottom=54pt, left=54pt, right=54pt]{geometry}
\newcommand{\changemargins}{
    \newgeometry{top=54pt, bottom=54pt, left=54pt, right=54pt}
}
\AtBeginDocument{\afterpage{\changemargins}}
\usepackage[ruled, vlined, linesnumbered]{algorithm2e}
\usepackage{algorithm2e}
\let\oldnl\nl
\newcommand{\nonl}{\renewcommand{\nl}{\let\nl\oldnl}}
\IEEEoverridecommandlockouts
\usepackage{cite}
\usepackage{amsmath,amssymb,amsfonts,amsthm}
\usepackage{algorithmic}
\usepackage{graphicx}
\usepackage{caption}
\usepackage{subcaption}
\usepackage{textcomp}
\usepackage{xcolor}
\usepackage{array, multirow}
\newcommand{\specialcell}[1]{\begin{tabular}{@{}c@{}}#1\end{tabular}}
\def\BibTeX{{\rm B\kern-.05em{\sc i\kern-.025em b}\kern-.08em
    T\kern-.1667em\lower.7ex\hbox{E}\kern-.125emX}}
\usepackage[colorlinks=black, 
            linkcolor=black, 
            urlcolor=blue, 
            bookmarksopen=true]{hyperref}

\newcommand{\bm}{\boldsymbol}

\newcommand{\mc}[1]{\mathcal{#1}}
\newtheorem{theorem}{Theorem}

\newtheorem{lemma}{Lemma}

\newtheorem{proposition}{Proposition}

\newtheorem{remark}{Remark}

\begin{document}
\title{Distributionally Robust Path Integral Control}
\author{\IEEEauthorblockN{Hyuk Park\textsuperscript{*}
\thanks{*Department of Industrial and Enterprise Systems Engineering, University of Illinois Urbana-Champaign. 
\texttt{\{hyukp2, duozhou2, gah\}}\texttt{@illinois.}\texttt{edu}
}
}
\and
\IEEEauthorblockN{Duo Zhou\textsuperscript{*}}
\and
\IEEEauthorblockN{Grani A. Hanasusanto\textsuperscript{*}}
\and
\IEEEauthorblockN{Takashi Tanaka\textsuperscript{\textdagger}
\thanks{\textdagger Department of Aerospace Engineering and Engineering Mechanics, University of Texas at Austin. 
\texttt{ttanaka@utexas.edu}}   
}
}

\maketitle
\begin{abstract}
We consider a continuous-time continuous-space stochastic optimal control problem, where the controller lacks exact knowledge of the underlying diffusion process, relying instead on a finite set of historical disturbance trajectories.
In situations where data collection is limited, the controller synthesized from empirical data may exhibit poor performance.
To address this issue, we introduce a novel approach named Distributionally Robust Path Integral (DRPI). The proposed method employs distributionally robust optimization (DRO) to robustify the resulting policy against the unknown diffusion process. Notably, the DRPI scheme shows similarities with risk-sensitive control, which enables us to utilize the path integral control (PIC) framework as an efficient solution scheme. We derive theoretical performance guarantees for the DRPI scheme, which closely aligns with selecting a risk parameter in risk-sensitive control. We validate the efficacy of our scheme and showcase its superiority when compared to risk-neutral PIC policies in the absence of the true diffusion process.
\end{abstract}


\section{Introduction}
We consider a continuous-time and continuous-space stochastic control problem where the explicit representation of the system dynamics is unknown to the controller. In numerous real-world applications, the absence of true system dynamics is a common issue, as the system dynamics may be too complex to model, or the collection of historical data for the true dynamics may be limited. One approach to address unknown dynamics is constructing a simulator based on available data to serve as a proxy for the true environment and to test control policies in the simulator before real-world deployment.

However, when the available data is sparse, the simulator may not accurately capture the true characteristics of the system dynamics. In other words, the probability distribution of the disturbance in the simulator model may not faithfully represent the distribution of the true disturbance in the real-world system. Such a distributional mismatch may cause the simulator-based control policy to exhibit poor performance in the true system. To address this issue, we propose a distributionally robust (DR) control problem, using the emerging distributionally robust optimization (DRO) paradigm. 
The proposed approach constructs an ambiguity set that contains all possible distributions from the nominal distribution within a certain distance. In the literature, various measures have been used to define the distance between two distributions~\cite{bayraksan2015data},~\cite{delage2010distributionally},~\cite{mehrotra2014models}. In this paper, we utilize the Kullback-Leibler (KL) divergence~\cite{kullback1951information} for the ambiguity set. Subsequently, the DR control problem seeks a policy that performs optimally under the worst-case disturbance distribution taken from the KL divergence-based ambiguity set. 
Interestingly, our DR control problem shows similarities with the traditional risk-sensitive control problem~\cite{whittle1990risk}. While several papers have drawn a connection between risk-sensitive control and DR control in discrete-time settings~\cite{petersen2000minimax},~\cite{nishimura2021rat}, there has not been prior work demonstrating this equivalence in continuous-time settings.

 The similarities between the two problems allow us to efficiently solve the DR control problem using the path integral control (PIC) method.
 PIC has emerged as a promising stochastic optimal control framework for designing feedback control systems in the literature~\cite{kappen2005path},~\cite{theodorou2010generalized}. It has recently found applications in diverse robotics domains, ranging from tasks such as aggressive driving~\cite{williams2016aggressive},~\cite{williams2017information}, off-road navigation~\cite{cai2022probabilistic} to multi-quadrotor control~\cite{williams2017model}, and more~\cite{patil2022chance},~\cite{patil2023simulator}. The key concept of PIC involves using the Feynman-Kac lemma~\cite{friedman1975stochastic} to convert the value function of the stochastic control problem to an expectation on all possible uncontrolled trajectories. This transformation allows for the approximation of the optimal control input using sample trajectories of uncontrolled system dynamics generated through the Monte Carlo method.
 
 PIC offers several advantages over traditional optimal control techniques. One notable advantage is its capability to handle arbitrary nonlinear state-dependent cost functions and state-dependent dynamics. This flexibility makes PIC well-suited for systems with complex behaviors that are difficult to capture using traditional methods. Moreover, it is less susceptible to high-dimensional state spaces or a long time horizon compared to other gradient-based methods. In practice, to expedite the computation of the optimal control, PIC can also harness modern graphics processing units (GPU) which can generate a large number of sample trajectories in parallel~\cite{williams2017model}. This parallelization significantly speeds up the computation process, enhancing the practical applicability of PIC for complex systems. 
 
 Recently, several works have proposed robust model predictive path integral control frameworks utilizing diverse methods, including covariance steering~\cite{yin2022trajectory}, conditional value-at-risk~\cite{yin2023risk}, control barrier functions~\cite{yin2023shield}, and robust sampling~\cite{williams2018robust},~\cite{gandhi2021robust}. However, to the best of our knowledge, our work is the first to combine PIC with the DRO methodology. The contributions of this paper are as follows: 
 \begin{enumerate}
     \item We formulate a continuous-time and continuous-space stochastic control problem as the DR control problem and establish its equivalence with risk-sensitive control and H-infinity control.
     \item We provide finite sample guarantees for the DR control policy under certain conditions. Our theoretical results serve as valuable guidance for the selection of the robust parameter.
     \item  In pursuit of an efficient solution scheme for the DR control problem, we propose a path integral-based algorithm and demonstrate its superior performance via numerical experiments in an online context.
 \end{enumerate}

The remaining sections of this paper are organized as follows. In Section~\ref{sec:problem_statement}, we define the stochastic control problem of our interest. In Section~\ref{sec:DRO}, we present our DR control problem and its reformulation as a single-level minimization problem. In Section~\ref{sec:data_driven}, we introduce the data-driven approach to the DR control problem and analyze the theoretical performance guarantees. In Section~\ref{sec:solution}, the solution scheme based on path integral control is presented. In Section~\ref{sec:experiment}, we present numerical experiments to demonstrate the effectiveness of our scheme. 
We conclude the paper with Section~\ref{sec:conclusion}, which summarizes the contributions of this work and suggests avenues for future research.

\subsection*{Notation}
Bold letters represent vectors and matrices, while regular fonts indicate scalars. 
the identity matrix is denoted as $\boldsymbol{\mathrm{I}}$—its dimension will be evident from the context. The partial derivatives with respect to the state $\bm{x}$ and time $t$ are denoted by $\partial_x$ and $\partial_t$, respectively. The trace operation of a square matrix $\bm{A}$ is denoted as $\operatorname{tr}(A)$. 
$N$ sequences, each containing $K$ vectors, are denoted as $\{\bm x^{(i)}(k)\}^K_{k=1}$, $\forall i=1,2,\ldots,N$.

\section{Problem Statement}\label{sec:problem_statement}
We consider a continuous-time and continuous-space stochastic dynamic $\bm x(t)\in\mathbb{R}^{n}$ that is affine in control and disturbance as follows:  
\begin{align}\label{eq:dynamic}
    \begin{split}
    d\bm x(t) = & f\left(\bm x(t),t\right)dt + \bm G\left(\bm x(t),t\right)
    \bm u(\bm x(t),t)dt + \\  & \bm \Sigma(\bm x(t), t) d\bm \xi(t).
    \end{split}
\end{align}
Here, $f\left(\bm x(t),t\right)\in \mathbb{R}^n$ is an arbitrary passive dynamic, $\bm u(\bm x(t),t)\in \mathbb{R}^{k}$ is a control input, $\bm G(\bm x(t),t)\in\mathbb{R}^{n\times k}$ is a full-rank control transition matrix function with $n\geq k$.
In addition, $\bm \xi(t) \in \mathbb{R}^{p}$ is a diffusion process on a suitable probability space $(\Omega, \mathcal{F}, \mathbb{P})$, and $\bm \Sigma(\bm x(t), t) \in \mathbb{R}^{n\times p}$ is a full-rank diffusion matrix function that maps the disturbance to the state. We define the diffusion process $\bm \xi(t)$ adapted to the filtration $\mathcal{F}_t$ as 
\begin{equation}\label{eq:diffusion_process}
d\bm \xi(t) = \bm \mu(t) dt + d\bm w.
\end{equation}
Here, $\bm \mu(t)$ represents the drift, and $d\bm w=\{\boldsymbol{\mathrm{w}}(t): t\geq 0\}$ is a standard Brownian disturbance with respect to the probability law $\mathbb{P}^\star$.
In this paper, we make the assumption that the controller lacks access to the true diffusion process \eqref{eq:diffusion_process} and $\mathbb{P}^\star$.
The assumption aligns with many real-world applications where accurately characterizing the actual disturbance is challenging.

Given a finite time horizon $t\in[0,T]$, the initial state $\bm x(0)$, a running cost function $\mathcal{L}_{\bm u}(\bm x(t), t)$, and a terminal cost function $\psi(\bm x(T), T)$, the stochastic optimal control problem is formulated as 
\begin{equation}\label{eq:true_cost}
\inf _{\bm u} 
\mathbb{E}_{\mathbb{P}^\star}
\Big[ 
\mc J_{\bm u} \left( \bm x(0), 0 \right)
\Big], 
\end{equation}
where $\mc J_{\bm u}(\cdot)=\psi\left(\bm x(T), T\right)+\int_{0}^T \mathcal{L}_{\bm u}\left(\bm x(s), s\right) d s$.
Here, the dependence of $\mc L_{\bm u}(\cdot)$ and $\mc J_{\bm u}(\cdot)$ on $\bm u$ implies a certain policy $\bm u(\bm x(t), t)$ imposed for evaluating the costs and $\mathbb{E}_{\mathbb{P}^\star}[\cdot]$ is the expectation evaluated under $\mathbb{P}^\star$. 
We assume that the running cost function $\mathcal{L}_{\bm u}(\cdot)$ consists of an arbitrary state-dependent cost $q(\cdot)$ and a quadratic control cost with a positive definite weight matrix $\bm R\in\mathbb{R}^{p\times p}$, given by
\begin{align}\label{eq:running_cost}
\mathcal{L}_{\bm u}\left(\bm x(t), t\right)=q\left(\bm x(t), t\right)+ \frac{1}{2}\bm u(\bm x(t), t)^{\top} \bm R \bm u(\bm x(t), t).
\end{align}

\section{Distributionally Robust Control Problem}\label{sec:DRO}
Designing the optimal policy for the stochastic control problem \eqref{eq:true_cost} clearly requires knowledge of the true diffusion process \eqref{eq:diffusion_process}, which is unrealistic in many cases. Instead, we consider a scenario where the controller has access to an alternative nominal diffusion process $\widehat{\bm \xi}(t)$, denoted as $d\widehat{\bm \xi}(t) = \widehat{\bm \mu}(t) dt + d\bm w$, where ${\bm w}$ is a Brownian disturbance under the probability law $\mathbb{Q}$. In this case, a na\"ive approach to designing a control policy might involve constructing the optimal controller  based on the dynamic system where the true diffusion process \eqref{eq:diffusion_process} is replaced by the nominal diffusion process. Nonetheless, this approach leads to suboptimal performance if the nominal diffusion process fails to accurately represent the true system dynamics. In fact, this issue is common 
 for simulator-based control policies, as the policies synthesized in erroneous simulator models often exhibit inferior performance when implemented in the actual system. 

To address this issue, we adopt the emerging paradigm of distributionally robust optimization (DRO) to formulate the distributionally robust (DR) control problem for the true problem~\eqref{eq:true_cost} as follows:
\begin{align}\label{eq:dr_control}
\inf _{\bm u} \sup_{\mathbb{P}\in{P}^{\gamma}\left(\mathbb{Q}\right)} 
\mathbb{E}_\mathbb{P}\Big[
\mc J_{\bm u} \left( \bm x(0), 0 \right)
\Big].
\end{align}
Here, the distributional (ambiguity) set~${P}^{\gamma}(\mathbb{Q})$ with robustness parameter~$\gamma>0$ is defined as
\begin{align}\label{eq:ambiguity_set}
\begin{split}
    {P}^{\gamma}\left(\mathbb{Q}\right) = \bigg\{&\mathbb{P}\in D\; :\; 
\mathbb{D}\left(\mathbb{P}\|\mathbb{Q}\right)=\\& \int_{\Xi}\log \frac{d\mathbb{P}}{d\mathbb{Q}}(\bm \xi) d\mathbb{P}(\bm \xi) \leq \gamma 
\bigg\}.
\end{split}
\end{align}
Here, ${D}$ denotes the set of all probability laws of $\bm \xi(\cdot)$ and $\mathbb{D}(\mathbb{P}\|\mathbb{Q})$ denotes the Kullback-Leibler (KL) divergence from $\mathbb{P}$ to $\mathbb{Q}$ where $d\mathbb{P}/d\mathbb{Q}$ is the likelihood ratio between $\mathbb{P}$ and $\mathbb{Q}$, also known as the Radon-Nikodym derivative. 

The DR control problem \eqref{eq:dr_control} seeks a policy that performs best under the worst-case probability law $\mathbb{P}$  within the ambiguity set, thus providing robustness against the unknown true diffusion process \eqref{eq:diffusion_process}. 
Compared to classical robust control, which is designed to optimize against worst-case disturbances, the DR control policy is less conservative, resulting in better performance across various applications, including robotics~\cite{nishimura2021rat}, control design~\cite{van2015distributionally}, and power systems~\cite{kim2020minimax}, etc. Furthermore, the DR control framework is well-suited for data-driven settings where a nominal diffusion process can be constructed based on available data. We will discuss how to construct the DR control policy in a data-driven manner in Section~\ref{sec:data_driven}.

\subsection*{Tractable Reformulation}\label{sec:DRO_A}
The min-max problem~\eqref{eq:dr_control} is inherently difficult to solve since the cost function involves a maximization problem. To design an efficient solution scheme, we introduce an equivalent single-level problem for the DR control problem~\eqref{eq:dr_control} by following the standard results of the convex analysis in~\cite{luenberger1997optimization} and duality between relative entropy and free energy in~\cite{petersen2000minimax},~\cite{dupuis2011weak}. 
We first make the following assumption 
\begin{equation}\label{assumption}
\sup_{\mathbb{P}\in D} 
\mathbb{E}_\mathbb{P}\Big[
\mc J_{\bm u} \left( \bm x(0), 0 \right)
\Big]=\infty.
\end{equation}
As discussed in~\cite{petersen2000minimax}, this assumption states that, without the KL divergence constraint, some arbitrary diffusion process can drive the expected cost to infinity. With this assumption,
we present the single-level reformulation for the DR control problem in the following lemma.
\begin{lemma}
For any given $\gamma >0$, the DR control problem~\eqref{eq:dr_control} can be equivalently reformulated as the following single-level minimization problem
\begin{align}\label{eq:dr_control_reformulation}
     \inf _{\bm u,\theta>0}  {\gamma}{\theta} + 
{\theta}\log\mathbb{E}_{\mathbb{Q}}\left[\exp\left(\frac{1}{\theta}
\mc J_{\bm u} \left( \bm x(0), 0 \right) \right)
\right].
\end{align}
\label{lem:reformulation}
\end{lemma}
\begin{proof}
Dualizing the inner maximization problem in \eqref{eq:dr_control} with the ambiguity set \eqref{eq:ambiguity_set} as a constraint, we have
\begin{align}
\small
&\sup_{\mathbb{P}\in D} \; 
\inf _{\theta \geq 0} \;\mathbb{E}_{\mathbb{P}}\Big[
\mc J_{\bm u} \left( \cdot \right)
\Big] - {\theta}\mathbb{D}\left(\mathbb{P}\|\mathbb{Q}\right)
+ {\gamma}{\theta}
\label{eq:primal_problem}
\\ 
&\leq \inf _{\theta \geq 0} {\gamma}{\theta} +  \sup_{\mathbb{P}\in D} \;
\mathbb{E}_{\mathbb{P}}\Big[
\mc J_{\bm u} \left( \cdot \right)
\Big] - {\theta}\mathbb{D}\left(\mathbb{P}\|\mathbb{Q}\right)
\label{eq:dual_problem}
\\ 
& = \inf _{\theta > 0} {\gamma}{\theta} +  \sup_{\mathbb{P}\in D} \;
\mathbb{E}_{\mathbb{P}}\Big[
\mc J_{\bm u} \left( \cdot \right)
\Big] - {\theta}\mathbb{D}\left(\mathbb{P}\|\mathbb{Q}\right).
\label{eq:reduced_dual_problem}
\end{align}
 Here, the inequality between \eqref{eq:primal_problem} and \eqref{eq:dual_problem} holds due to weak duality, and the equality between \eqref{eq:dual_problem} and \eqref{eq:reduced_dual_problem} holds since the objective function in \eqref{eq:dual_problem} evaluates to infinity for $\theta=0$ by the assumption \eqref{assumption}.
 
 To derive the single-level problem~\eqref{eq:dr_control_reformulation}, we first establish strong duality (i.e., the equality) between \eqref{eq:primal_problem} and \eqref{eq:dual_problem} by following ~\cite[Theorem 1, Chapter 8]{luenberger1997optimization}.
 It is clear that $D$ in \eqref{eq:primal_problem} is a convex set and the objective function in \eqref{eq:primal_problem} is concave in $\mathbb{P}$ since $\mathbb{E}_{\mathbb{P}}[\mc J_{\bm u}(\cdot)]$ is linear in $\mathbb{P}$ and $\mathbb{D}\left(\mathbb{P}\|\mathbb{Q}\right)$ is convex in $\mathbb{P}$. 
 Furthermore, we can show the existence of an interior point $\mathbb{P}\in D$: Let $\mathbb{P}=\mathbb{Q}$, then the following strict inequality holds:
 \begin{equation*}    \mathbb{D}\left(\mathbb{Q}\|\mathbb{Q}\right)= \int_{\Xi}\log \frac{d\mathbb{Q}}{d\mathbb{Q}}(\bm \xi) d\mathbb{P}(\bm \xi)=0 < \gamma\
 \end{equation*}
 for any $\gamma > 0$. 
 Hence, by incorporating the minimization over $\bm u$ with \eqref{eq:reduced_dual_problem}, the DR control problem \eqref{eq:dr_control} becomes equivalent to
 \begin{equation*}    
 \inf _{\bm u,\theta > 0} {\gamma}{\theta} +  \sup_{\mathbb{P}\in D} \;
\mathbb{E}_{\mathbb{P}}\Big[
\mc J_{\bm u} \left( \cdot \right)
\Big] - {\theta}\mathbb{D}\left(\mathbb{P}\|\mathbb{Q}\right).
\label{eq:dr_reformulation_2}
 \end{equation*}

Then, the remainder of the proof amounts to showing 
\begin{align}\label{eq:lagrende_duality}
    \sup_{\mathbb{P}\in D} \;
\mathbb{E}_{\mathbb{P}}\Big[
\mc J_{\bm u} \left( \cdot \right)
\Big] - {\theta}\mathbb{D}\left(\mathbb{P}\|\mathbb{Q}\right)
=
{\theta}\log\mathbb{E}_{\mathbb{Q}}\left[\exp\left(\frac{1}{\theta}
\mc J_{\bm u} \left( \cdot \right) \right)
\right]
\end{align}
for all $\theta > 0$.
This equality \eqref{eq:lagrende_duality}, known as the Legendre duality between the KL divergence and free energy, is already shown in \cite[Section 4.6.3]{dupuis2011weak}.
This completes the proof.
\end{proof}

Note that the reformulation~\eqref{eq:dr_control_reformulation} exhibits similarities with the risk-sensitive control problem~\cite{whittle1990risk}. In fact, if we consider $\theta$ as a parameter, the problem~\eqref{eq:dr_control_reformulation} is equivalent to the risk-sensitive control problem. This shows the intimate relationship between DR control and risk-sensitive control. A low value of the robustness parameter~$\gamma$ in $\eqref{eq:dr_control_reformulation}$ favors a high value of $\theta$. 
In particular, as $\gamma$ converges to 0, the problem becomes equivalent to the risk-neutral control problem. 
Conversely, a higher value of $\gamma$ yields a more risk-averse policy. 

While the connection between risk-sensitive control and DR control has been shown in several papers~\cite{petersen2000minimax},~\cite{nishimura2021rat}, they have focused on discrete-time settings. More importantly, it remains an open question on how to determine the risk-sensitive parameter $\theta$. 
In the literature, the selection of $\theta$ often involves a trial-and-error process where policies with different values of $\theta$ are tested in the actual environment until the resulting policy aligns with the modeler's risk profile. 
However, such experimentation can be costly in safety-critical tasks. 
On the other hand, as will be discussed in Section~\ref{sec:data_driven}, 
our DR control framework offers theoretical performance guarantees in data-driven settings, guiding the selection of the robustness parameter $\gamma$ based on available data. This choice corresponds to an appropriately selected $\theta$ in the risk-sensitive control framework.

\section{Data-Driven Approach}\label{sec:data_driven}
In this section, we explore the application of the DR control framework in a data-driven context. Specially, we discuss the construction of the approximate diffusion process and the corresponding probability law $\mathbb{Q}$ using available data. 
We introduce an additional assumption in this section: the drift in the true diffusion process is deterministic, i.e.,
\begin{equation}\label{eq:true_diffusion_deterministic}
d\bm \xi = \bm \mu dt + d\bm w.     
\end{equation}
While the DR control problem \eqref{eq:dr_control} discussed in Section~\ref{sec:DRO} can accommodate a more general $\mc F_t$-adapted diffusion process \eqref{eq:diffusion_process}, the assumption \eqref{eq:true_diffusion_deterministic} enables us to construct an approximate diffusion process in a data-driven manner and establish theoretical performance guarantees for the DR control policy, as we will discuss in this section.

In the absence of the true diffusion process 
\eqref{eq:true_diffusion_deterministic}, we assume that the controller merely has access to 
$N$ sequences of historical disturbance terms collected over time interval $\Delta t>0$ denoted as $\{\Delta \bm \xi^{(i)}(k\Delta t)\}^K_{k=1}, \; \forall i=1,2,\ldots,N$, where $K=T/\Delta t$ is the number of empirical disturbances for each sequence. 
Subsequently, we can construct the approximate diffusion process
\begin{equation}\label{eq:approx_diffusion}
d\widehat{\bm \xi} = \widehat{\bm \mu}dt+ d\bm w, \text{ where } \widehat{\bm \mu}=\frac{1}{NK\Delta t}\sum_{i=1}^N \sum_{k=1}^{K}\Delta \bm \xi^{(i)}(k\Delta t)
\end{equation}
and $d\bm w$ is a Brownian disturbance under the probability law $\mathbb{Q}$. 
Consequently, we can employ \eqref{eq:approx_diffusion} as a nominal diffusion process in the ambiguity set ${P}^{\gamma}(\mathbb{Q})$ defined in \eqref{eq:ambiguity_set}, rendering a data-driven DR control model.

\subsection*{Performance Guarantees}\label{sec:OOS}
We provide the finite sample guarantee of our DR control policy. The ambiguity set $P^\gamma(\mathbb{Q})$ centered at $\mathbb{Q}$ of the approximate diffusion process \eqref{eq:approx_diffusion} can be viewed as a random object in a sense that different realizations of the empirical disturbance terms $\{\Delta \bm \xi^{(i)}(k\Delta t)\}^K_{k=1}, \; \forall i=1,2,\ldots,N$, may result in a different approximate diffusion process. 
Intuitively, more data would provide a more reliable estimate of the true diffusion process \eqref{eq:true_diffusion_deterministic}. We
can establish the following generalization bound by leveraging measure concentration theory and Girsanov's theorem~\cite{oksendal2013stochastic}.
\begin{proposition}\label{lem:UBS}
  Consider $\mathbb{P}^{\star}$ and $\mathbb{Q}$ as the probability laws of the true and approximate diffusion processes, given by \eqref{eq:true_diffusion_deterministic} and \eqref{eq:approx_diffusion}, respectively. 
  Then, for any given value of the robust parameter $\gamma$ 
, we have
 \begin{equation} \label{eq:OOS_bound} 
\mathbb{P}^{\star}\in{P}^{\gamma}(\mathbb{Q}) \text{ w.p. at least } 1-2 p \exp \left(-\frac{\gamma N}{\sqrt{p}}\right),
\end{equation}
where $p$ is the dimension of the diffusion process.
\label{lem:UB}
\end{proposition}
\begin{proof}
Suppose $\mathbb{P}^{\star}\in{P}^{\gamma}(\mathbb{Q})$ for any given value of $\gamma$. This implies 
\begin{equation}\label{eq:KL_div}
\mathbb{D}\left(\mathbb{P}^{\star}\|\mathbb{Q} \right)=\int_{\Xi}\log \frac{d\mathbb{P}^\star}{d\mathbb{Q}}(\bm \xi) d\mathbb{P}^\star(\bm \xi)\leq \gamma.     
\end{equation}
Applying Girsanov's Theorem ~\cite[ Theorem 8.6.5]{oksendal2013stochastic}, the KL divergence in \eqref{eq:KL_div} is equivalent to 
\begin{equation*}
    \mathbb{D}\left(\mathbb{P}^{\star}\|\mathbb{Q} \right)=\mathbb{E}_{\mathbb{P}^{\star}}\left[-\int_{0}^T \boldsymbol{\mu}^{\top} d \boldsymbol{w}+\frac{1}{2}\int_{0}^T\|\widehat{\boldsymbol{\mu}}-\boldsymbol{\mu}\|^2 d s\right],
\end{equation*}
where $\widehat{\bm \mu}$ is defined in \eqref{eq:approx_diffusion}. Since $\bm w\in\mathbb{R}^p$ is a Brownian disturbance under $\mathbb{P}^\star$, we can rewrite  \eqref{eq:KL_div} as 
\begin{equation}\label{eq:OOS_1}
\frac{T}{2}\|\widehat{\boldsymbol{\mu}}-\boldsymbol{\mu}\|^2\leq \gamma.     
\end{equation}
For any time step size $\Delta t > 0$, we know that the empirical disturbances $\Delta\bm\xi^{(i)}(k\Delta t)$, for all $i=1,2,\ldots,N$, and $k=1,2,\ldots,K$, are i.i.d. Gaussian samples drawn from $\mathcal{N}(\Delta t\bm\mu, \Delta t\boldsymbol{\mathrm{I}})$. Therefore, using Hoeffding's inequality\cite{hoeffding1963probability} and union bound, we have
\begin{equation*}
\begin{aligned}
&\operatorname{Pr}\left[\frac{T}{2}\|\widehat{\boldsymbol{\mu}}-\boldsymbol{\mu}\|^2 \leq \gamma\right] \geq 1-2 p \exp \left(-\frac{\gamma  N  }{\sqrt{p}}\right) .
\end{aligned}
\end{equation*}
This completes the proof.
\end{proof}

Proposition \ref{lem:UB} provides a lower bound on the probability that the unknown $\mathbb{P}^\star$ is contained in $P^\gamma(\mathbb{Q})$ for a fixed $\gamma$. Similarly, for any fixed $\epsilon\in (0,1)$, we can obtain the smallest value of $\gamma(\epsilon)$ that guarantees $\mathbb{P}^{\star}\in{P}^{\gamma}(\mathbb{Q})$ w.p. $1-\epsilon$ : 
by setting $\epsilon=2 p \exp \left(-\gamma N/\sqrt{p}\right)$ in \eqref{eq:OOS_bound}, we can rewrite \eqref{eq:OOS_bound} as 
\begin{equation}
    \mathbb{D}\left(\mathbb{P}^{\star}\|\mathbb{Q} \right)\leq \gamma(\epsilon)=\frac{\sqrt{p}}{N}\log\left(\frac{2p}{\epsilon}\right) 
    \label{eq:finite_sample}
\end{equation}
w.p. at least $1-\epsilon$.

If $\mathbb{P}^\star$ generated by the true stochastic process \eqref{eq:true_diffusion_deterministic} is indeed contained in the ambiguity set, the optimal value of the DR control problem serves as an upper bound on the true cost of implementing the DR control policy in the real environment. Finally, using \eqref{eq:finite_sample}, we propose the following finite sample guarantee for our DR control policy.
\begin{theorem}\label{thm:finite}
Suppose that $\widehat{\mc J}^N$ and $\widehat{\bm u}^N(\bm x(t),t),\; \forall t\in[0,T]$, represent the optimal value and the optimal policy of the DR control problem \eqref{eq:dr_control} with ambiguity set $P^{\gamma(\epsilon)}(\mathbb{Q})$ where we set $\gamma(\epsilon)={\sqrt{p}}/{N}\log\left({2p}/{\epsilon}\right)$ for a fixed value of $\epsilon\in(0,1)$ as in \eqref{eq:finite_sample}. Then,
we have
\begin{equation} \label{eq:guarantee}
\mathbb{E}_{\mathbb{P}^\star}
\Big[ 
\mc J_{\widehat{\bm u}^N} \left( \bm x(0), 0 \right)
\Big]  
     \leq \widehat{\mc J}^N  
\end{equation}
w.p. at least $1-\epsilon$.
\end{theorem}
\begin{proof}
The claim immediately holds from \eqref{eq:finite_sample} since
\begin{equation*}
\mathbb{E}_{\mathbb{P}^\star}
\Big[ 
\mc J_{\widehat{\bm u}^N} \left( \bm x(0), 0 \right)
\Big]
     \leq
      \inf _{\bm u} \sup_{\mathbb{P}\in{P}^{\gamma}\left(\mathbb{Q}\right)} 
\mathbb{E}_{\mathbb{P}}
\Big[ 
\mc J_{{\bm u}} \left( \bm x(0), 0 \right)
\Big],
 \end{equation*}
 whenever $\mathbb{P}^\star\in P^\gamma(\mathbb{Q})$.
 \end{proof}
  Theorem \ref{thm:finite} provides valuable guidance for modelers on choosing $\gamma(\epsilon)$ that guarantees a prescribed confidence level $1-\epsilon$ prior to real-world implementation of the control policy. Furthermore, given that $\gamma$ is at most $\mathcal{O}(1/N)$ as shown in \eqref{eq:finite_sample}, we can adjust its value at a rate of $1/N$ as more empirical data becomes available.
  \begin{remark}
  The equivalence between risk-sensitive control and H-infinity control in the linear quadratic Gaussian (LQG) setting was initially noted in \cite{jacobson1973optimal}. However, they highlighted the equivalence between two Riccati equations from the two different controls without any consideration of relative entropy. As demonstrated in the proof of Proposition~\ref{lem:UBS}, Girsanov's theorem transforms the ambiguity set~\eqref{eq:ambiguity_set} into the uncertainty set~\eqref{eq:OOS_1} concerning the unknown drift $\bm \mu$, bounded by the robustness parameter $\gamma$. This implies the equivalence between our DR control and the nonlinear generalization of H-infinity control, hence, expanding upon the earlier observation. 
  \end{remark}

\section{Solution Scheme}\label{sec:solution}
\subsection{Decomposition}
Using the equivalence between the risk-sensitive control and the DR control shown in Section~\ref{sec:DRO_A}, we can decompose the reformulated DR control problem~\eqref{eq:dr_control_reformulation} into two minimization problems as follows: the master problem is
\begin{equation}
    \label{eq:master_prolem}
    \begin{array}{ccll}
    \displaystyle \inf_{\theta > 0} {\gamma}{\theta} + g(\theta),
    \end{array}
\end{equation}
and the subproblem is
\begin{equation}
\label{eq:subproblem}
g\left(\theta\right)  =   \inf _{\bm u}
{\theta}\log\mathbb{E}_{\mathbb{Q}}\left[
\exp\Big(
\frac{1}{\theta} \mc J_{\bm u} \left( \bm x(0), 0 \right)
\Big)
\right].
    \end{equation}
Note that the master problem \eqref{eq:master_prolem} is merely an univariate optimization problem over $\theta>0$ which can be solved by various methods. Then, for a fixed $\theta$, the subproblem \eqref{eq:subproblem} becomes the risk-sensitive control problem where the expectation of the exponentiated cost function $\mc J_{\bm u}(\cdot)$ is evaluated under $\mathbb{Q}$.

\subsection{Risk-Sensitive Path Integral Control}
To efficiently solve the subproblem \eqref{eq:subproblem}, we utilize the risk-sensitive path integral control framework proposed in \cite{broek2012risk}. They demonstrated that the standard (i.e., risk-neutral) path integral method~\cite{kappen2005path} can be generalized to risk-sensitive control under the same assumption. In this section, we briefly restate their main results for clarity and completeness. Further details and derivations can be found in \cite{broek2012risk}.

Solving the risk-sensitive control problem \eqref{eq:subproblem} involves setting up the following second-order partial differential equation (PDE) known as the stochastic Hamilton-Jacobi-Bellman (HJB) equation  
\begin{align}
\label{eq:HJB}
    \begin{split}
        -{\partial_t \mc V^{\theta}(\bm x(t),t)} &= \inf _{\bm u} 
     \Bigg( \partial_x\mc V^{\theta\top}(f+\bm G\bm u+\bm\Sigma\widehat{\bm\mu})  
    \\ & + \frac{1}{2} \operatorname{tr} \left( \partial_{\bm x}^2 \mc V^{\theta}
    \bm \Sigma \bm \Sigma^\top 
    \right) + q + \frac{1}{2}\bm u^{\top} \bm R \bm u
    \\ &  
     + \frac{1}{2\theta} \Big\Vert 
     \bm \Sigma^\top \partial_{\bm x} \mc V^{\theta} 
     \Big\Vert^2 \Bigg) (\bm x(t), t)
    \end{split}
\end{align}
with boundary condition $\mc V^{\theta}(\bm x(T),T)=\psi(\bm x(T), T)$.
Taking derivative with respect to $\bm u$ on the right-hand side in \eqref{eq:HJB}, one can derive the optimal control 
\begin{equation}\label{eq:optimal_control}
\bm u (\bm x(t), t) = - \bm R^{-1} \bm G(\bm x(t),t)^\top {\partial_{\bm x} \mc V^{\theta}(\bm x(t), t)}. 
\end{equation}
Substituting \eqref{eq:optimal_control} into \eqref{eq:HJB}, we have
\begin{align}\label{eq:HJB_2}
\begin{split}
    &-\partial_t \mathcal{V}^{\theta}(\bm{x}(t),t) = \Bigg( \partial_x \mathcal{V}^{\theta\top}(f+\bm{\Sigma}\widehat{\bm{\mu}})  + q \\
    &\quad + \frac{1}{2} \operatorname{tr} \bigg( \partial_{\bm{x}}^2 \mathcal{V}^{\theta}\left(
\frac{1}{\theta}\bm{\Sigma} \bm{\Sigma}^\top - \bm{G}\bm{R}^{-1} \bm{G}^\top \right) \bigg)
 \Bigg) (\bm{x}(t), t).
\end{split}
\end{align}
Note that the HJB equation \eqref{eq:HJB_2} is generally nonlinear in $\mc V^{\theta}(\cdot)$ due to the last term on the right-hand side of \eqref{eq:HJB_2}. The conventional solution method is to solve the HJB equation backward in time over the entire time horizon $[0,T]$ for all $\bm x(t)$. This recursive backward evaluation suffers from the curse of dimensionality, becoming intractable as the dimension of the state space increases. 

Risk-sensitive path integral control can be used as an alternative solution approach to the backward recursion for the special case of \eqref{eq:HJB_2} where there exists ${\theta}^*$ that satisfies the equation
\begin{align}\label{eq:lambda_equality}
\theta^*\bm G(\bm x,t)\bm R^{-1}\bm G(\bm x,t)^\top=\bm \Sigma(\bm x,t)\bm \Sigma(\bm x,t)^{\top}.
\end{align}
Note that the equality \eqref{eq:lambda_equality} implies that the HJB equation \eqref{eq:HJB_2} becomes linearizable since the nonlinear term in $\mc V^{\theta}(\cdot)$ disappears when $\theta^*$ is used in \eqref{eq:HJB_2}. 
In a one-dimensional case, \eqref{eq:lambda_equality} holds trivially, while, in a higher dimensional space, it may impose constraints on the choice of the control dependent cost matrix $\bm R$.  Therefore, in this paper, we make the assumption that $\theta$ satisfying \eqref{eq:lambda_equality} always exists.
 If a chosen $\theta$ for the subproblem \eqref{eq:subproblem} satisfies \eqref{eq:lambda_equality}, the HJB equation is immediately linear. 
 
 If a chosen $\theta$ does not satisfy \eqref{eq:lambda_equality}, we can utilize a log transformation of the value function, defined as:
\begin{equation}
\label{eq:exp_transformation}
\mc V^{\theta}(\bm x(t),t) = -\frac{1}{\theta'} \log\left(\Psi(\bm x(t), t)\right),
\end{equation}
where $\theta'=\theta/(\theta^*\theta-1)$.
This transformation renders the HJB equation linear in terms of $\Psi(\cdot)$. Subsequently, relying on the well-known Feynman-Kac lemma~\cite{friedman1975stochastic}, we can derive the solution to the linearized HJB equation using the path integral control framework with the dynamic \eqref{eq:dynamic} being affine in control and disturbance, and the quadratic control cost in \eqref{eq:true_cost} as defined earlier.
This framework enables the computation of the value function using all possible forward trajectories of the \textit{uncontrolled} dynamics, i.e., $ d\bm x(t) = f(\bm x(t),t)dt + \bm \Sigma(\bm x(t),t)d\bm \xi(t)$, as follows:
\begin{align}
\label{eq:forward_trajectory}
\begin{split}
    \mc V^{\theta}(\bm x(t),&t) = {\theta'} \log \mathbb{E}_{\mathbb{Q}}\left[
\exp\left(\frac{1}{\theta'} \mc J_{\bm 0} \left(\bm x(t), t \right) \right)
\right],
\end{split}
\end{align}
where $\mc J_{\bm 0} \left(\cdot\right)=\psi\left(\bm x(T)\right) +
\int_{t}^T q\left(\bm x(s)\right)ds$ represents the cost of an uncontrolled trajectory.
Furthermore, as demonstrated in~\cite{theodorou2015nonlinear}, taking the derivative of $\Psi(\cdot)$ with respect to $\bm x$ yields the optimal control for the problem \eqref{eq:true_cost} at time $t$, as follows:
\begin{align}
\label{eq:optimal_control_cont}
\begin{split}
    &\bm u\left(\bm x(t),t\right)dt=\\ &{\textstyle \bm R^{-1}\bm G_{c}^\top (\bm G_{c} \bm R^{-1} \bm G_{c}^\top)^{-1}\left(\bm x(t), t\right)\frac{\mathbb{E}_{\mathbb{Q}}\left[\exp\left({\lambda'}\mc J_{\bm 0}  \right)\bm \Sigma_{c} d\bm w(t)  \right]}{\mathbb{E}_{\mathbb{Q}}\left[\exp\left({\lambda'} \mc J_{\bm 0} \right)\right]}.}
\end{split}
\end{align}
Here, $\bm G_{c}(\cdot)\in\mathbb{R}^{(n-l)\times p}$ and $\bm \Sigma_{c}(\cdot)\in\mathbb{R}^{(n-l)\times k}$ represent submatrices of the control transition and diffusion matrices $\bm G(\cdot)$ and $\bm \Sigma(\cdot)$, respectively. These submatrices correspond to the directly actuated states denoted as $\bm x_c(t)\in\mathbb{R}^l$ within the state vector $\bm x(t)=[\bm x_c(t), \;\bm x_p(t)]^\top$ where $l\leq n$ without loss of generality. The remaining part of the states, $\bm x_p(t)$, represents non-directly actuated states.
\begin{algorithm}[!ht]
\linespread{1}\selectfont
\DontPrintSemicolon
\SetKwInOut{Init}{Initialization}
\KwIn{ $\bm x(0)$: Initial state;\\
    $f(\cdot),\bm G(\cdot),\bm \Sigma(\cdot)$: System dynamics;\\
    $\bm G_c(\cdot)/\bm \Sigma_c(\cdot)$: Submatrix for control transition/diffusion matrix;\\
    $q(\cdot),\psi(\cdot)/\bm R$: State/Control cost;\\
    $\widehat{\bm \mu}(\cdot)$: Estimated drift term;\\
    $\gamma$: Robustness parameter;\\
    $M$: Number of sample trajectories;\\
    $\Delta t$: Step size}
\KwOut{Control input ${\bm u}(k)$ for $k=0,1\ldots,K-1$}
\SetKwBlock{Begin}{function}{end function}
\For{$k \gets 0$ to $K-1$}
{ Sample $M$ trajectories of disturbance $\{\bm \varepsilon^{(i)}(\tau)\}^M_{i=1}$ for $\tau=k,k+1,\ldots,K-1$  \\
\For {$ i\gets 1$ to $M$}{
    Initialize the $i$-th cost $\mc J^{(i)}_{\bm 0} \gets 0$\\
  \For{$ k'\gets k$ to $K-1$} {
  $\bm x^{(i)}(k'+1)=\bm x^{(i)}(k') + f(\bm x^{(i)}(k'),k')\Delta t
  +  \bm \Sigma (\bm x^{(i)}(k'),k') ( \widehat{\bm \mu}(\bm x^{(i)}(k'),k')\Delta t 
  + \bm \varepsilon^{(i)}(k')\sqrt{\Delta t}$ ) \\
    \If{$k' < K-1$}
     {$\mc J^{(i)}_{\bm 0} += q(\bm x^{(i)}(k'+1),k'+1)$} 
     \Else
     {$\mc J^{(i)}_{\bm 0} += \psi\left(\bm x^{(i)}(K),K)\right)$}  
  }
  }
  {
  $\displaystyle \widehat{\theta} \in \operatorname*{arg\,min}_{\theta > 0} {\gamma}{\theta} 
  + {\theta'}  \log\left(\frac{1}{M}\sum_{i=1}^{M} 
  \exp\left( \frac{\mc J^{(i)}_{\bm 0}}{\theta'}  \right)
  \right)
  $\label{alg:line11}
  \\
  \For {$ i\gets 1$ to $M$}{
   $\displaystyle r^{(i)} = \exp\left( \frac{\mc J^{(i)}_{\bm 0}}{\widehat{\theta}}  \right)\bigg/ \sum_{j=1}^M \exp\left( \frac{\mc J^{(j)}_{\bm 0} }{\widehat{\theta}} \right)$ 
   }
   $\displaystyle \bm u_k \gets \bm R^{-1}\bm G_{c}^\top (\bm G_{c} \bm R^{-1} \bm G_{c}^\top)^{-1}(\bm x(k),k)\sum_{i=1}^{M} r^{(i)} \frac{\bm  \Sigma_c \bm \varepsilon^{(i)}(k)}{\sqrt{\Delta t}}$ \\
  Send ${\bm u}(k)$ to the controller \\
  Update the current state ${\bm x}(k+1)$  
  }
  }
\caption{DRPI}\label{alg:DRPI}
\end{algorithm}
\subsection{Numerical Method}
The numerical implementation of the PIC requires two types of approximation, namely, time discretization and sampling trajectories. First,
we can approximate the continuous-time dynamic system \eqref{eq:dynamic} using the Euler-Maruyama method~\cite{kloeden1992stochastic}: 
\begin{align}
\label{eq:time_discretized_dynamic}
\begin{split}
\bm x(k+1)= &\bm x(k) + f(\bm x(k),k)\Delta t+ \bm G(\bm x(k),k)\bm u(\bm x(k),k)\Delta t\\ & 
+ \bm\Sigma(\bm x(k),k)\Delta\bm \xi(k)
\end{split}
\nonumber
\end{align}
for $k=0,1,\ldots,K-1$, where time step size $\Delta t>0$ determines the total number of time steps, i.e., $K= T /\Delta t$ and $\Delta\bm \xi(k)={\bm\mu}(k)\Delta t + \bm \varepsilon\sqrt{\Delta t}$ is the discrete-time diffusion process where $\bm \varepsilon\sim\mc N(\bm 0, \boldsymbol{\mathrm{I}})$.  
Subsequently, we can approximately compute the optimal control for the discrete-time version as we estimate the expectation in \eqref{eq:optimal_control_cont}. This approximation involves utilizing a collection of $M$ uncontrolled sample trajectories $\{\bm x^{(i)}(k)\}^{K }_{k=1}$ for $i=1,\ldots,M$ generated via the Monte-Carlo method as follows:
\begin{align*}
 &{\bm u}(\bm x(k),k) = \\& {\textstyle  \bm R^{-1}\bm G_{c}^\top (\bm G_{c} \bm R^{-1} \bm G_{c}^\top)^{-1}(\bm x(k),k)\frac{\sum_{i=1}^M \exp\left({\lambda'} \mc J^{(i)}_{\bm 0}\right) \bm \Sigma_c \frac{\bm \varepsilon}{\sqrt{\Delta t}}}{\sum_{j=1}^M \exp\left({\lambda'}\mc J^{(j)}_{\bm 0}\right)},}
\end{align*}
where $\mc J^{(i)}_{\bm 0} = \psi(\bm x^{(i)}(K),K) + \sum_{s=k}^{K-1} q(\bm x^{(i)}(s),s)\Delta t$ is a cost associated with the $i$-th trajectory $\{\bm x^{(i)}(s)\}_{s=k}^K$.

The proposed Distributionally Robust Path Integral (DRPI) Algorithm \ref{alg:DRPI} makes use of the risk-sensitive path integral control to compute the optimal value of the subproblem $g(\theta)$~\eqref{eq:subproblem} for any $\theta>0$. For the master problem~\eqref{eq:master_prolem}, we can employ various line search methods in line \ref{alg:line11} since it is an univariate optimization. Note that, in line \ref{alg:line11}, we reuse the costs $\mc J^{(i)}_{\bm 0}$ for $i=1,2,\ldots,M$, to optimize over $\lambda$ without sampling new trajectories. Hence, the scheme efficiently solves the master problem.
\begin{remark}
For the special case where the stochastic control problem is convex, $g(\theta)$ becomes convex since it is a partial minimization of the convex function over $\bm u$. 
In this case, we can solve the master problem in polynomial time by a binary search over $\theta$ while $g(\theta)$ is solved by the path integral method for each $\theta$.
If we further simplify the problem to the case where the dynamic is linear and the cost function is quadratic in both $\bm x$ and $\bm u$, the subproblem $g(\theta)$ becomes the well-known linear exponential-of-quadratic Gaussian (LEQG) control problem, which can be analytically solved by the modified backward Riccati equation~\cite{jacobson1973optimal}.
\end{remark}
\section{Numerical Experiments}\label{sec:experiment}
In this section, we demonstrate the effectiveness of our DRPI scheme. To reflect limited data availability, we start each simulation run with an initial estimate $\widehat{\bm \mu}$ using only a single data point $\Delta \bm \xi^{(1)}$. 
Here we set the maximum time step to  $K$. Then, for each time step  $k=1, 2, \ldots, K$, we refine $\widehat{\bm\mu}$ as we have access to a newly observed empirical disturbance term $\Delta \bm \xi^{(k)}$, as follows:
\begin{equation*}\label{eq:online}
\widehat{\bm\mu} \leftarrow \widehat{\bm\mu}  + \frac{\Delta \bm\xi^{(k)}/\Delta t  - \widehat{\bm\mu} }{k}, k=1,2,\ldots, K.
\end{equation*}
As our estimate $\widehat{\bm\mu}$ improves with the increasing number of available data, we adjust $\gamma \leftarrow  1/k$ for each time step $k=1,2,\ldots, K$, based on the theoretical results in Section~\ref{sec:OOS}.
We compare our DRPI control with the risk-neutral path integral control (PIC), i.e., DRPI with $\gamma$  fixed at 0 throughout the experiments:
as mentioned in Section~\ref{sec:DRO_A}, if $\gamma=0$, the ambiguity set \eqref{eq:ambiguity_set} only includes the empirical distribution, making the inner maximization disappear in \eqref{eq:dr_control}.
\subsection{Double Integrator Model}
Consider the following double integrator model for a particle robot in a 2D plane as follows:
\begin{equation}\label{eq:numerical_dynamic}
\begin{split}
\small
     \bm x(k+1) =\begin{bmatrix}  {p}_{x}(k+1) \\ {p}_{y}(k+1) \\ {h}(k+1) \\ {v}(k+1)  \end{bmatrix}& =  \bm x(k) +  \begin{bmatrix} 0 & 0 & 1 & 0 \\ 0 & 0 & 0 & 1 \\ 0 & 0 & 0 & 0 \\ 0 & 0 & 0 & 0 \end{bmatrix} \bm x(k)  \Delta t
    \\ 
    & + \begin{bmatrix} 0 & 0 \\ 0 & 0 \\ 1 & 0 \\ 0 & 1\end{bmatrix} \begin{bmatrix}
        u_h(k)\Delta t + \Delta\xi_h(k) \\ u_v(k)\Delta t + \Delta\xi_v(k)
    \end{bmatrix}.
\end{split}
\end{equation}
Here, the states ${p}_{x}(k)$ and ${p}_{y}(k)$ represent positions in the horizontal and vertical directions, while $h(k)$ and $v(k)$ represent linear velocities in the corresponding directions. The control inputs $u_h(k)$ and $u_v(k)$ denote the acceleration in the horizontal and the vertical direction, respectively. The true disturbance $\bm \xi(k) = [\xi_h(k), \: \xi_v(k)]^\top$ is a Brownian disturbance with unknown drift.
Given initial state $\bm x(0)=[-3.5,\; 2.5,\; 0.0,\; 0.0]^\top$, the objective is to design an optimal control policy $\bm u^*(\bm x(k),k)=[u^{*}_h(\bm x(k),k),\; u^{*}_v(\bm x(k),k)]^\top$ for the robot to arrive at the target position $[p_{x}^{*}(k), \; p_{y}^{*}(k)]=[0.0,\; 0.0]$ as soon as possible, while avoiding any collisions with both the inner square obstacle and the outer square boundary as shown in Figure \ref{fig:DI}(a).
\begin{figure}[htbp]
    \centering
    \begin{minipage}{0.24\textwidth} \label{fig:DIM}
        \centering
        \includegraphics[width=\textwidth]{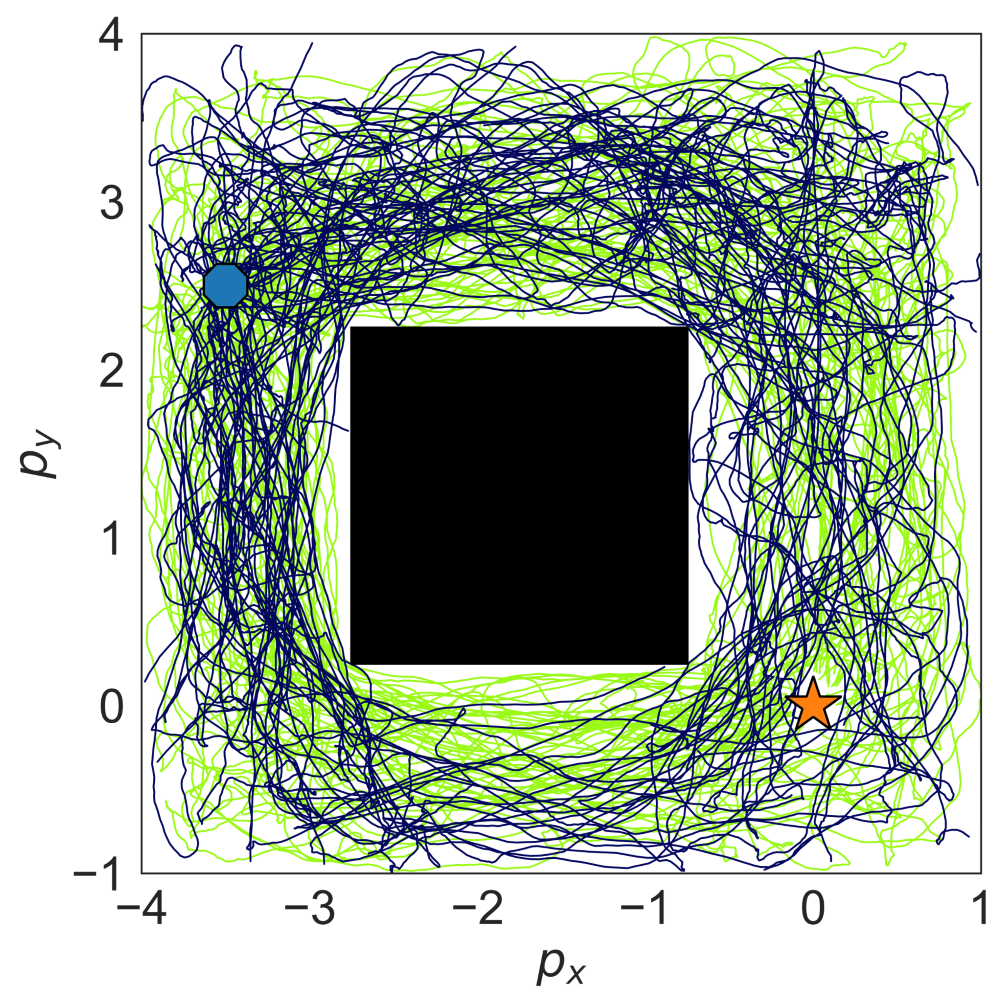}
        {\small \it PIC}
        \smallskip 
    \end{minipage}
    \hfill
    \begin{minipage}{0.24\textwidth}
        \centering
        \includegraphics[width=\textwidth]{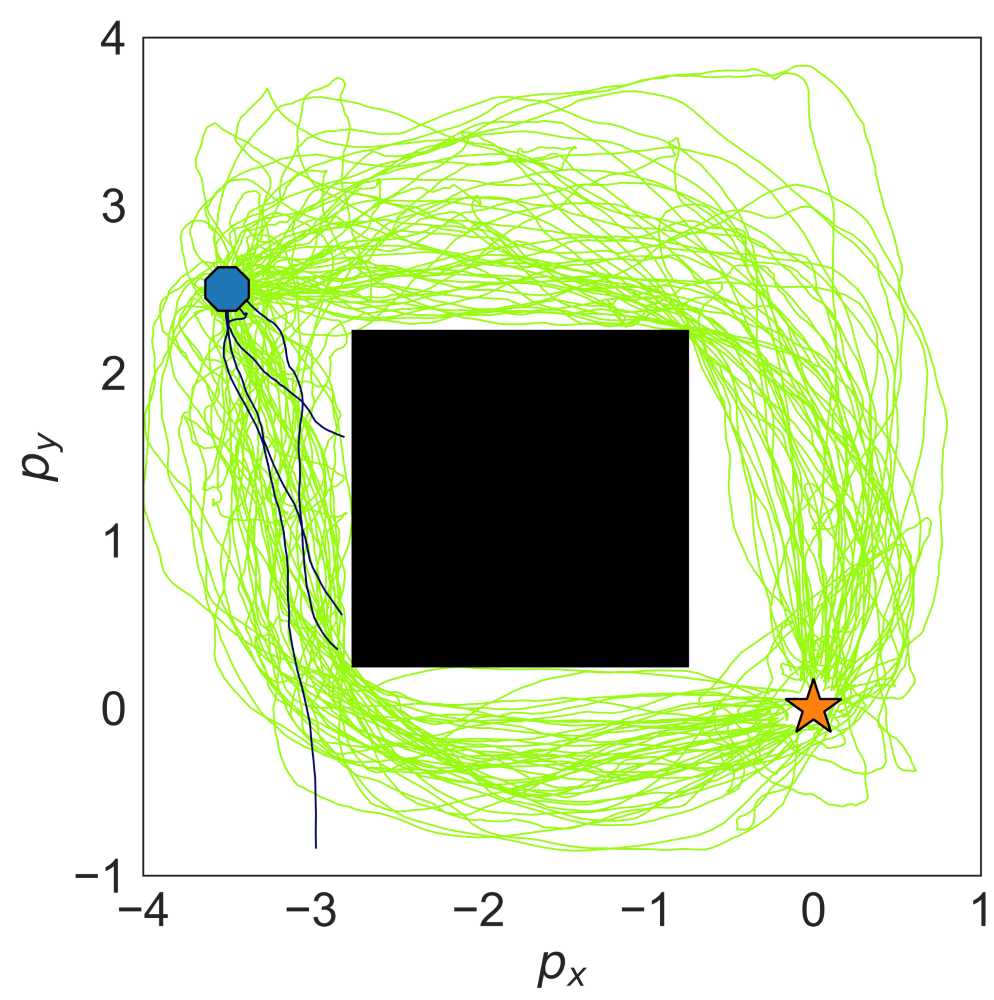}
        {\small \it DRPI}
        \smallskip 
    \end{minipage}
    Fig. 1(a) double integrator model
    \vspace{1em} 

    \begin{minipage}{0.24\textwidth}\label{fig:Uni}
        \centering
        \includegraphics[width=\textwidth]{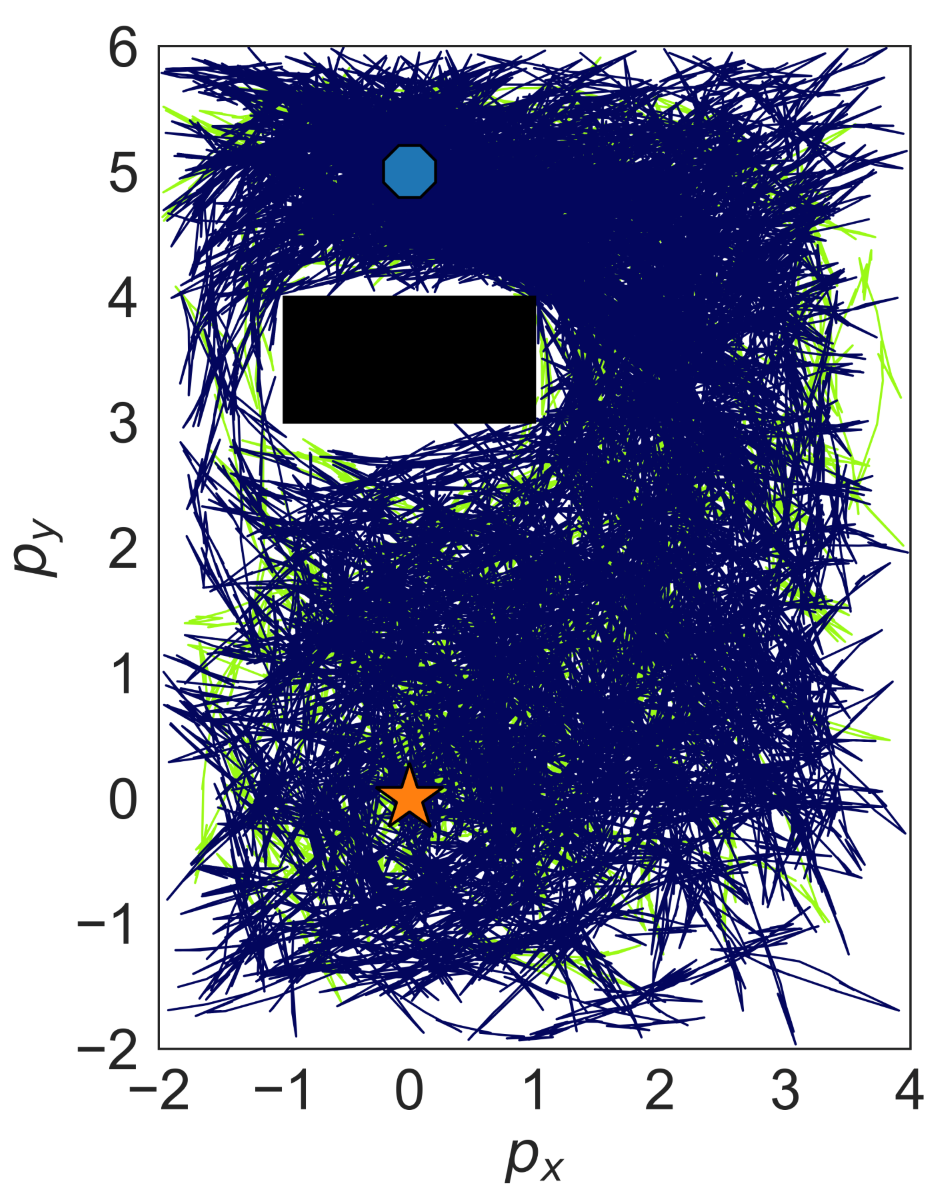}
        {\small \it PIC}
        \smallskip 
    \end{minipage}
    \hfill
    \begin{minipage}{0.24\textwidth}
        \centering
        \includegraphics[width=\textwidth]{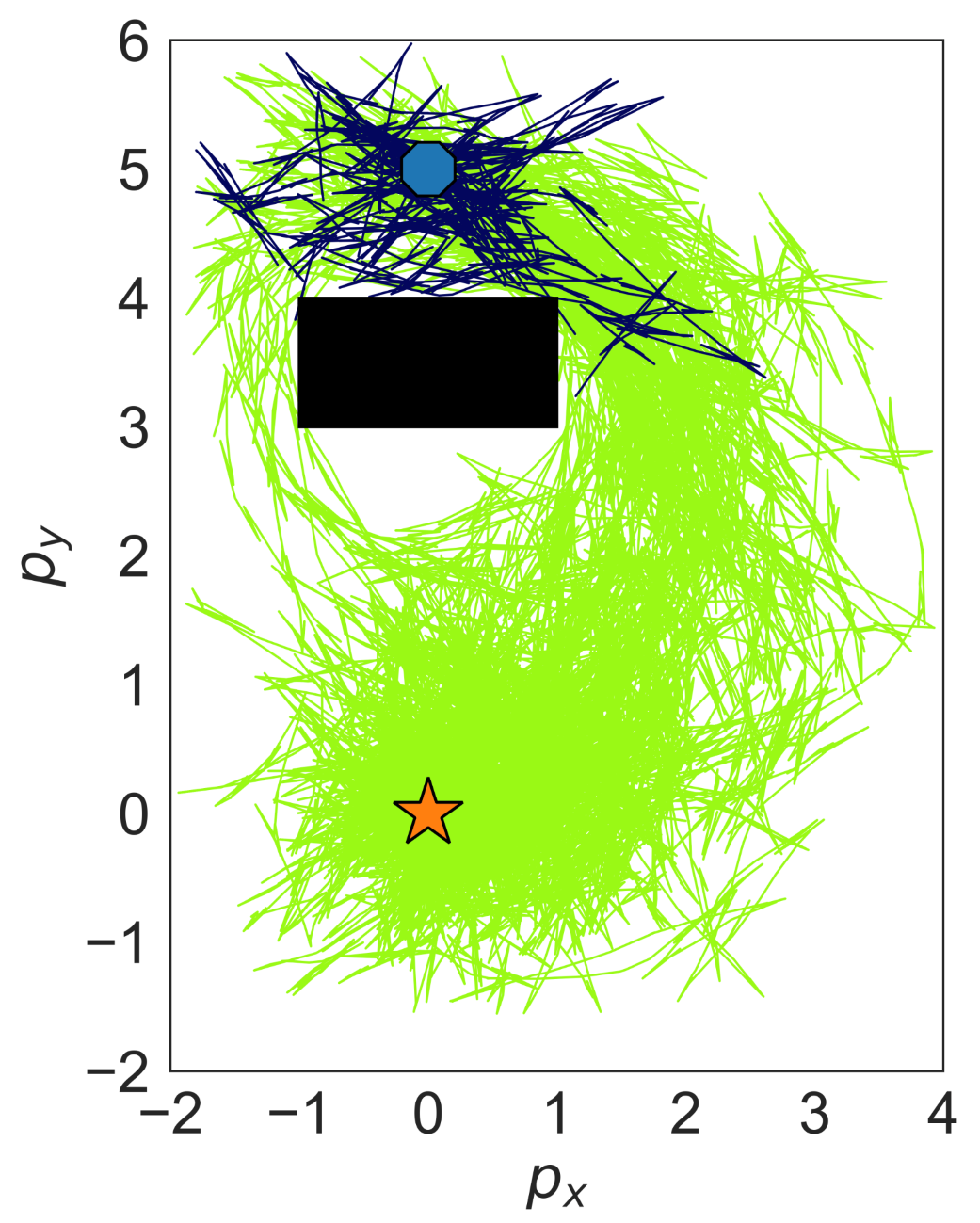}
        {\small \it DRPI}
        \smallskip 
    \end{minipage}
    Fig. 1(b) unicycle model

    \caption{In the context of robot navigation, two models are considered: (a) the double integrator model and (b) the unicycle model. The initial positions are represented by blue octagons, while the target positions are denoted by orange stars. The figure illustrates 100 trajectory simulations generated using the PIC schemes (left) and DRPI (right) for each model, under the online context with $\widehat{\bm\mu} \leftarrow \Delta \bm \xi^{(1)}$, which means $\widehat{\bm\mu}$ is estimated by a single initial data point, and initial $\gamma$ value of $1$ for each simulation. Trajectories are color-coded; deep blue trajectories collide with either inner or outer boundaries, whereas green trajectories represent successful simulations.}
    \label{fig:DI}
\end{figure}

\subsection{Unicycle Model}
The second model is an unicycle model of a robot that can move forward and change its orientation in the 2D plane, but cannot move directly sideways as follows:
\begin{equation}
\label{eq:unicycle_dynamic}
\begin{split}
\small
    \bm{x}(k+1) &= \begin{bmatrix}  {p}_{x}(k+1) \\ {p}_{y}(k+1) \\ {p}_{\theta}(k+1) \end{bmatrix} \\
    &= \bm{x}(k) + \begin{bmatrix} \cos(\theta) & 0 \\ \sin(\theta) & 0 \\ 0 & 1 \end{bmatrix} \begin{bmatrix}
        u_v(k)\Delta t + \Delta\xi_v(k) \\ u_{\omega}(k)\Delta t + \Delta\xi_{\omega}(k)
    \end{bmatrix}.
\end{split}
\end{equation}
Here, the states ${p}_{x}(k)$ and ${p}_{y}(k)$ represent positions in the horizontal and vertical direction, respectively, and ${p}_{\theta}(k)$ represent the orientation of the unicycle. Control inputs $u_v(k)$ and $u_{\omega}(k)$ represent forward velocity and the rate of change of the orientation, respectively. 

Similar to the first experiment, the true Brownian disturbance $\bm \xi(k) = [\xi_v(k), \: \xi_{\omega}(k)]^\top$ is unknown and the objective is to design an optimal control policy $\bm u^*(\bm x(k),k)=[u^{*}_v(\bm x(k),k),\; u^{*}_{\omega}(\bm x(k),k)]^\top$ for the robot to navigate to the target position $[p_{x}^{}(k), \; p_{y}^{}(k)]=[0.0,\; 0.0]$, without any collision as shown in Figure \ref{fig:DI}(b).

\subsection{Experimental Design}
 For both the double integrator and unicycle models, we consider quadratic control-dependent cost with the weight matrix $\bm R=10^{-3}\boldsymbol{\mathrm{I}}\in\mathbb{R}^{2\times 2}$ and a nonlinear state-dependent cost function $q(\bm x(k))$ as follows:
\begin{equation}\label{eq:double_integrator_cost}
    q\left(\bm x(k)\right)=c_1 \cdot \Vert \bm x(k) - \bm x^{*}\Vert_{2} + c_2 \cdot q_{_{o}}\left(\bm x(k)\right) + c_3 \cdot q_{_{b}}\left(\bm x(k)\right).
\end{equation}
Here, the obstacle cost $q_{{o}}(\cdot)$ and the boundary cost $q{_{b}}(\cdot)$ are indicator functions that take the value 1 if the robot hits the obstacle and the boundary, respectively. The coefficient parameters in equation \eqref{eq:double_integrator_cost} are set to $c_1 = 10^{-2}$, $c_2=c_3 = 10^2$. In the event of a collision with the inner or outer squares, the current simulation is immediately terminated and recorded as a failure. 

We conducted 100 simulations for each model using both DRPI and the PIC schemes, and recorded the trajectory for each simulation and the arrival times for the successful ones. The results summarized in Table \ref{tab1} and Figure \ref{fig:DI} show that the DRPI scheme offers significant advantages over the PIC scheme across both models, which include a higher success rate, faster arrive time on average, and more consistent and predictable performance due to lower variability. 

\begin{table}[htbp]
\caption{Performance of Different Schemes}
\begin{center}
\begin{tabular}{|c|c|c|c|c|c|}
\hline
\multirow{2}{*}{\textbf{Model}}  & \multirow{2}{*}{\textbf{Scheme}} & \multirow{2}{*}{\specialcell{\textbf{Success}\\\textbf{Rate (\%)}}}  & \multicolumn{3}{|c|}{\textbf{Arrive Time (s)}} \\
\cline{4-6} 
 & & & \textbf{Mean} & \textbf{Std. Dev.} & \textbf{95 pct.} \\
\hline
\multirow{2}{*}{\specialcell{Double\\Integrator}}  & PIC & 66 & 21.30 & 12.63 & 44.42 \\
 & DRPI & 96 & 8.04 & 2.25 & 12.32 \\
\hline
\multirow{2}{*}{\specialcell{Unicycle}}  & PIC & {19} & {25.89} & {11.13} & {44.52} \\
 & DRPI & {78} & {16.94} & {6.17} & {28.30} \\
\hline
\end{tabular}
\label{tab1}
\end{center}
\end{table}

\section{Conclusion}\label{sec:conclusion}
We introduced a novel data-driven approach to robustify policies for a broad class of stochastic control problems in the absence of the true diffusion process. Notably, we established an interesting connection between the DR control and the risk-sensitive control. Our theoretical results offered valuable insight into selecting the robust parameter, making simulator-based controllers more practical when dealing with unknown system dynamics. Furthermore, our proposed DRPI algorithm showcased outstanding performance. Future research avenues
may include exploring extensions of our framework to more
complex systems in various application domains.
\bibliographystyle{IEEEtran}
\bibliography{main}
\end{document}